\documentclass[conference]{IEEEtran}
\IEEEoverridecommandlockouts
\usepackage{
amsmath,cite,bm,amsfonts,amssymb,
graphicx,color,mathtools,setspace,
comment,multirow,xfrac,epstopdf,
footnote,multirow,lscape,float,
caption,subcaption,amsthm,alltt,placeins,amsthm}

\newtheorem{theorem}{Theorem}

\newcommand{\myincludegraphics}
{\includegraphics[trim=0.5cm 0cm 0.5cm 0.7cm, clip=true, width=1\columnwidth]}
\newcommand{\raisecapt}{\vspace{-0.1cm}}
\newcommand{\myVM}[3]{\mathbf{#1}_{\mathrm{#2}}^{#3}} 
\newcommand{\Norm}[3]{\left\lVert #1 \right\rVert_{#2}^{#3}} 
\newcommand{\tr}[3]{\mathrm{tr}\left\{\mathbf{#1}_{\mathrm{#2}}^{#3}\right\}}
\newcommand{\eq}[1]{(#1)}
\newcommand{\Diag}[1]{\mathrm{diag}\{#1\}}
\newcommand{\Abs}[2]{\left\lvert #1 \right\rvert^{#2}}
\newcommand{\Exp}[1]{\mathbb{E}\left\{ #1 \right\}}

\newcommand{\betad}[0]{\beta_{\mathrm{d}}}
\newcommand{\betabr}[0]{\beta_{\mathrm{br}}}
\newcommand{\betaru}[0]{\beta_{\mathrm{ru}}}

\newcommand{\ab}[1]{\myVM{a}{b}{#1}}
\newcommand{\ar}[1]{\myVM{a}{r}{#1}}

\newcommand{\hru}[1]{\myVM{h}{ru}{#1}}
\newcommand{\hd}[1]{\myVM{h}{d}{#1}}
\newcommand{\Hbr}[1]{\myVM{H}{br}{#1}}

\begin{document}

\bstctlcite{IEEEexample:BSTcontrol} 

\title{Phase Dependent Loss Analysis for RIS Systems}

\author{\IEEEauthorblockN{%
		Ikram Singh\IEEEauthorrefmark{1}, %
		Peter J. Smith\IEEEauthorrefmark{2}, %
		Pawel A. Dmochowski\IEEEauthorrefmark{1}}
	\IEEEauthorblockA{\IEEEauthorrefmark{1}%
		School of Engineering and Computer Science, Victoria University of Wellington, Wellington, New Zealand}
	\IEEEauthorblockA{\IEEEauthorrefmark{2}%
		School of Mathematics and Statistics, Victoria University of Wellington, Wellington, New Zealand}
	\IEEEauthorblockA{email:%
		~\{ikram.singh,peter.smith,pawel.dmochowski\}@ecs.vuw.ac.nz
	}%
}

\maketitle
\begin{abstract}
In this paper we focus on phase dependent loss (PDL), an important aspect of reconfigurable intelligent surfaces (RIS)  where the signals reflected from the RIS elements are attenuated by varying amounts depending on the phase rotation provided by the element. To evaluate the effects of PDL,  we analyse the SNR of a SIMO RIS-aided wireless link. We assume that the channel between the base station (BS) and RIS is a rank-1 LOS channel while the user (UE)-BS and UE-RIS are correlated Rayleigh channels. The RIS design is optimal in the absence of PDL and  maximizes the SNR in this scenario. Specifically, we derive a closed form expression for the mean SNR in the presence of PDL.  The attenuation function used for PDL was developed from a detailed circuit analysis of RIS elements. Leveraging the derived results, we analytically characterise the impact of PDL on the mean SNR. Numerical results are conducted to validate the derived expressions and verify the analysis.
\end{abstract}
\IEEEpeerreviewmaketitle
%
%
\section{Introduction}\label{Sec: Introduction}
Research into reconfigurable intelligent surfaces (RISs) has shown that intelligently tuning the RIS phases can significantly  improve performance in wireless systems. However, such works usually assume that reflections from the RIS elements experience a  constant attenuation. This is an over-simplification and assumes that the power of the reflected signal is independent of the phase shift at each RIS element. In this paper, we focus on the more general case \cite{PracPhaseShift}, where the RIS phases affect the reflected signal strength, i.e., phase dependent loss (PDL). As an initial investigation, we focus on the effects of PDL on single user (SU) systems.

For SU systems, \cite{9095301} derives a closed form expression for the mean SNR where the user (UE) to RIS and RIS to base station (BS) channels experience Rayleigh fading and the direct channel between UE and BS is absent.
\cite{ISinghRayleigh} derives an exact expression for the optimal uplink (UL) mean SNR for systems where the UE-BS channel is rank-1 LOS and the UE-RIS and UE-BS channels are correlated Rayleigh. The LOS assumption in the RIS-BS channel has been considered and motivated in numerous works (e.g \cite{Max_Min}). The authors in  \cite{ISinghRayleigh} leverage the mean SNR expression to provide insight on the impact of correlation on the mean SNR.
In \cite{RayleighRicean}, the authors extend the exact mean SNR derivation in \cite{ISinghRayleigh} to systems where the UE-BS and UE-RIS channels are correlated Ricean and derive a tight approximation to the mean rate. The authors again, leverage the mean SNR expression to provide insight on the impact of correlation and the Rician K-factor on the mean SNR. 
However, the analysis in \cite{9095301,ISinghRayleigh,RayleighRicean} assumes either perfect RIS reflection or reflections with  constant attenuation.

In \cite{PracPhaseShift}, a mathematical model is proposed for PDL. Numerical results in \cite{PracPhaseShift} show that the  model accurately matches the reflective response of a detailed circuit model for a semiconductor device used to construct typical RIS elements. Furthermore,  characteristics of the circuit model resemble experimental results in the literature \cite{PracPhaseShift}.  

To best of our knowledge, no analysis is available to characterise optimal system performance with PDL. Hence, the contributions of this paper are as follows:
\begin{itemize}
    \item An exact mean SNR expression is derived for the optimal RIS phases using the PDL model in \cite{PracPhaseShift}. The optimal RIS design is based on the lossless case as there is no known optimal design in the presence of PDL. A simple rule of thumb is also provided to evaluate the effects of PDL.
    \item We analytically characterise the impact of the parameters in the loss function (attenuation function) on the mean SNR. The loss function is defined by three parameters; $L_{\mathrm{min}}$: the minimum amplitude of the loss function; $\alpha$: the steepness of the loss function; $\theta$: the shift of the loss function. These parameters are dependent on the circuit used to construct typical semiconductors for RIS reflective elements.
    \item Numerical results validate the derived SNR expression and verify the impact of the loss function parameters $L_{\mathrm{min}},\alpha,\theta$ on the mean SNR. We show that any impact caused by the loss function on the mean SNR becomes more pronounced as the size of the RIS  increases. For typical parameter values, these effects are significant.
\end{itemize}
\textit{Notation:} $\Exp{\cdot}$ represents statistical expectation. $\Re\left\{ \cdot \right\}$ is the Real operator. $\Norm{\cdot}{2}{}$ denotes the $\ell_{2}$ norm. Upper and lower boldface letters represent matrices and vectors, respectively. $\mathcal{CN}(\boldsymbol{\mu},\mathbf{Q})$ denotes a complex Gaussian distribution with mean $\boldsymbol{\mu}$ and covariance matrix $\mathbf{Q}$. $\mathcal{U}[a,b]$ denotes a uniform distribution on $[a,b]$. The transpose, Hermitian transpose and complex conjugate operators are denoted as $(\cdot)^{T},(\cdot)^{H},(\cdot)^{*}$, respectively. The trace and diagonal operators are denoted by $\tr{\cdot}{}{}$ and $\text{diag}\{\cdot\}$, respectively. The angle of a vector $\myVM{x}{}{}$ of length $N$ is defined as $\angle \myVM{x}{}{} = [ \angle x_{1},\ldots, \angle x_{N} ]^{T}$ along with $\Abs{\mathbf{x}}{} = [\Abs{x_1}{},\ldots,\Abs{x_N}{}]^T$. The exponential of a vector is defined as $e^{\myVM{x}{}{}} = [ e^{x_{1}},\ldots,  e^{x_{N}} ]^{T}$. $\otimes$ denotes the Kronecker product. $\mathrm{B}\left( z, w \right)$ denotes the beta function with parameters $z,w$. $\mathbf{1}_N$ denotes an $N \times 1$ vector with unit entries.

\section{System Model}\label{Sec: System Model}
As shown in Fig.~\ref{Fig: System Model}, we examine a RIS-aided single input multiple output (SIMO) system where a RIS with $N$ reflective elements is located close to a BS with $M$ antennas such that a rank-1 LOS condition is achieved between the RIS and BS.
\begin{figure}[h]
	\centering
	\includegraphics[width=5cm]{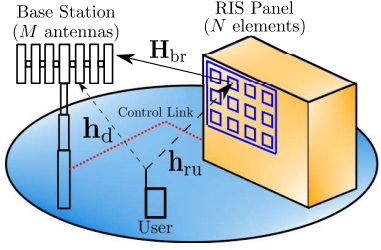}
	\raisecapt\caption{System model (the red dashed line is the control link for the RIS phases).}
	\label{Fig: System Model}
\end{figure}
\subsection{Channel Model}\label{Sec: Channel Model}
Let $\myVM{h}{d}{} \in \mathbb{C}^{M \times 1}$, $\myVM{h}{ru}{} \in \mathbb{C}^{N \times 1}$, $\myVM{H}{br}{} \in \mathbb{C}^{M \times N}$ be the UE-BS, UE-RIS and RIS-BS channels, respectively. The diagonal matrix $\myVM{\Phi}{}{} \in \mathbb{C}^{N \times N}$, where $\mathbf{\Phi}_{rr} = e^{j\phi_{r}}$ for $r=1,2,\ldots,N$, contains the reflection coefficients for each RIS element. The global UL channel is thus represented by
\begin{equation}\label{Eq: 1}
	\mathbf{h} = \myVM{h}{d}{} + \myVM{H}{br}{} \myVM{\Phi}{}{} \mathbf{L}\left( \myVM{\Phi}{}{} \right) \myVM{h}{ru}{},
\end{equation}
with $\mathbf{L}\left( \myVM{\Phi}{}{} \right) = \Diag{L(\phi_1),\ldots,L(\phi_N)}$, where the amplitude of the reflected signal at the $n^{\mathrm{th}}$ element is attenuated by the loss factor, $L(\phi_n) \in [0,1]$. Note that although the analysis in the paper is applicable to any loss function, we adopt the following practical loss model for RIS reflective elements based on detailed modeling of the RIS circuit elements in \cite{PracPhaseShift}
\begin{equation}\label{Eq: Phase Loss Function}
    L(\phi_n) = 
    (1 - L_{\mathrm{min}})\left( \frac{\sin(\phi_n + \theta) + 1}{2} \right)^\alpha + L_{\mathrm{min}}.
\end{equation}
The PDL model in \eqref{Eq: Phase Loss Function} gives losses which are dependent on the RIS phases. The variables $L_{\mathrm{min}} \geq 0, \theta \geq 0$ and $\alpha \geq 0$ are constants dependent on specific circuit implementations \cite{PracPhaseShift}. $L_{\mathrm{min}}$ controls the minimum amplitude of the loss function, $\alpha$ controls the steepness of the loss function and $\theta$ control the mid-point position of the loss function. Note that perfect RIS phase reflection can be achieved by setting $L_{\mathrm{min}} = 1$ or equivalently $\alpha = 0$.

For $ \myVM{h}{d}{}$ and  $\myVM{h}{ru}{}$, we assume correlated Rayleigh channels:
\begin{equation}\label{Eq: channels hd hru}
    \myVM{h}{d}{} = \sqrt{\beta_{\textrm{d}}} \myVM{R}{d}{1/2} \myVM{u}{d}{},
    \quad
    \myVM{h}{ru}{} = \sqrt{\beta_{\textrm{ru}}} \myVM{R}{ru}{1/2} \myVM{u}{ru}{},
\end{equation}
where $\beta_{\textrm{d}}$ and $\beta_{\textrm{ru}}$ are the link gains, $\myVM{R}{d}{}$ and $\myVM{R}{ru}{}$ are the correlation matrices for UE-BS and UE-RIS links respectively and $\myVM{u}{d}{},\myVM{u}{ru}{} \sim \mathcal{CN}(\mathbf{0},\mathbf{I})$. The rank-1 LOS channel from RIS to BS has link gain $\beta_{\textrm{br}}$ and is given by $\myVM{H}{br}{} = \sqrt{\beta_{\textrm{br}}}\myVM{a}{b}{}\myVM{a}{r}{H}$ where $\myVM{a}{b}{}$ and $\myVM{a}{r}{}$ are topology specific steering vectors at the BS and RIS respectively. Particular examples of steering vectors for a vertical uniform rectangular array (VURA) are in Sec.~\ref{Sec: Results}.

Note, that the correlation matrices, $\myVM{R}{ru}{}$ and $\myVM{R}{d}{}$, can represent any correlation model. For simulation purposes, we will use the well-known exponential decay model for correlation at the BS and adopt the sinc correlation model for correlation at the RIS \cite[Eq.~(11)]{RISCorrModel}. Hence,
\begin{equation}\label{Eq: Correlation Model}
\begin{split}
&(\myVM{R}{ru}{})_{ik} = \mathrm{sinc}(2 d_{i,k}) \text{ with } \mathrm{sinc}(2 d_r) = \rho_{\mathrm{ru}},
\\
&(\myVM{R}{d}{})_{ik} = \rho_{\mathrm{d}}^{\frac{d_{i,k}}{d_{\textrm{b}}}},
\end{split}
\end{equation}
where $0 \leq \Abs{\rho_{\text{ru}}}{} \leq 1$, $0 \leq \Abs{\rho_{\text{d}}}{} \leq 1$. $d_{i,k}$ is the distance between the $i^{\textrm{th}}$ and $k^{\textrm{th}}$ antenna/element at the BS/RIS. $d_{\textrm{b}}$ is the nearest-neighbour BS antenna separation measured in wavelength units. $\rho_{\mathrm{d}}$ and $\rho_{\mathrm{ru}}$ are the nearest neighbour BS antenna and RIS element correlations, respectively. Observe that the correlation model used at the RIS is a sinc function and the correlation level, $\rho_{\mathrm{ru}}$, is directly linked to the RIS element spacing $d_r$.

\subsection{Optimal RIS Matrix}\label{SubSec: Optimal RIS matrix}
Using  \eq{\ref{Eq: 1}}, the received signal at the BS is,
$
\mathbf{r} = \left(\myVM{h}{d}{} + \myVM{H}{br}{} \myVM{\Phi}{}{} \mathbf{L}\left( \myVM{\Phi}{}{} \right) \myVM{h}{ru}{} \right)s + \mathbf{n} \triangleq \myVM{h}{}{}s + \mathbf{n},
$
where $s$ is the transmitted signal with power $E_{s}$ and $\mathbf{n} \sim \mathcal{CN}(\mathbf{0},\sigma^{2}\mathbf{I})$. For a SU system, matched filtering (MF) is optimal, with UL SNR, given by
$
\text{SNR} = \bar{\tau}\Norm{\mathbf{h}}{2}{2},
$
where $\bar{\tau} = \frac{E_{s}}{\sigma^{2}}$. Thus, to maximize the SNR with lossless RIS reflection ($\mathbf{L}\left( \myVM{\Phi}{}{} \right) =\mathbf{I}$), the optimal RIS matrix is given by \cite[Eq.~(4)]{RayleighRicean},
\begin{equation}\label{Eq: Optimum PHI}
	\myVM{\Phi}{}{} = 
	\psi
	\Diag{e^{j\angle\myVM{a}{r}{}}}
	\Diag{e^{-j\angle\myVM{h}{ru}{}}},
\end{equation}
where $\psi = \frac{\myVM{a}{b}{H} \myVM{h}{d}{}}{\Abs{\myVM{a}{b}{H} \myVM{h}{d}{}}{}}$. Thus, the UL SNR is 
\begin{align}\label{Eq: SNR Eq}
	\text{SNR} &= 
	\bar{\tau} \Big( \hd{H}\hd{} + 
	2\Re\left\{ \hd{H}\Hbr{}\myVM{\Phi}{}{} \mathbf{L}\left( \myVM{\Phi}{}{} \right) \hru{} \right\} \notag \\
	& \quad + \hru{H} \mathbf{L}\left( \myVM{\Phi}{}{} \right) \mathbf{\Phi}^{H}
	\Hbr{H}\Hbr{} \myVM{\Phi}{}{} \mathbf{L}\left( \myVM{\Phi}{}{} \right) \hru{} \Big).
\end{align}
In this paper, we assume that the optimal lossless design in \eqref{Eq: Optimum PHI} is used in the presence of phase dependent loss. This is reasonable as an optimal design in the presence of loss is unknown. Note that in obtaining \eqref{Eq: SNR Eq}, we set $\mathbf{L}\left( \myVM{\Phi}{}{} \right)^H = \mathbf{L}\left( \myVM{\Phi}{}{} \right)$ since $\mathbf{L}\left( \myVM{\Phi}{}{} \right)$ is a positive real valued diagonal matrix.

\section{Mean SNR}\label{Sec: Mean SNR}
Here, we provide an exact result for the mean SNR, $\Exp{\mathrm{SNR}}$,
building on the results in \cite{ISinghRayleigh} for the mean SNR in a lossless scenario.

\begin{theorem}
The mean SNR is given by
\begin{align}
    \Exp{\mathrm{SNR}}
    &=
    \bar{\tau}\Big( \betad M + 
    \sqrt{\betabr\betad\betaru}\Norm{\myVM{R}{d}{1/2}\ab{}}{}{}N\mu_1\frac{\pi}{2} \notag \\
    & \quad + 
    \betaru \betabr M(N\mu_2 + F) \Big) \label{Eq: Mean SNR Final},
\end{align}
with
\begingroup
\allowdisplaybreaks
\begin{align}
    \mu_1 &= \frac{4^\alpha(1-L_{\mathrm{min}})}{\pi}\mathrm{B}\left( \frac{2\alpha+1}{2}, \frac{2\alpha+1}{2} \right) + L_{\mathrm{min}} \label{Eq: mu1 final}, \\
    \mu_2 &= \frac{2 L_{\mathrm{min}}(1 - L_{\mathrm{min}}) 4^{\alpha}}{\pi} \mathrm{B}\left( \frac{2\alpha+1}{2}, \frac{2\alpha+1}{2} \right) \notag \\
    &+
    L^2_{\mathrm{min}}
    +
    \frac{(1 - L_{\mathrm{min}})^2 16^\alpha}{\pi} \mathrm{B}\left( \frac{4\alpha+1}{2}, \frac{4\alpha+1}{2} \right) \label{Eq: mu2 final}, \\
    F &= \underset{r \neq s}{\sum_{r=1}^{N} \sum_{s=1}^{N}} \dfrac{\pi}{4}\left( 1 - \left\lvert\rho_{rs}\right\rvert^2\right)^2 {}_{2}F_{1}\left(\frac{3}{2},\frac{3}{2};1;\left\lvert\rho_{rs}\right\rvert^2 \right) L_{rs}, \label{Eq: F final}
\end{align}
\endgroup
where
\begingroup
\allowdisplaybreaks
\begin{align}
    L_{rs} &= \Exp{L(\phi_r)L(\phi_s)} \notag \\
    &= \int_{0}^{2\pi}
    \int_{0}^{2\pi}
    L(s+\angle(\ar{})_r-2\pi) L(t+\angle(\ar{})_s-2\pi) \notag \\
    & \quad \times g_{rs}(t-s) \ ds dt, \label{Eq: Lrs integral}\\
    g_{rs}(x) &=  \frac{1-\Abs{\rho_{rs}}{2}}{4 \pi^2} \left(
    \frac{1}{1 - v_{rs}(x)^2} - \frac{v_{rs}(x) \cos^{-1}(v_{rs}(x))}{(1 - v_{rs}(x)^2)^{3/2}}
    \right) \label{Eq: g_rs}, \\
    v_{rs}(x) &= \Abs{\rho_{rs}}{} \cos(x-\angle(-\rho_{rs})) \label{Eq: v_rs},
\end{align}
${}_{2}F_{1}(\cdot)$ is the Gaussian hypergeometric function, $\rho_{rs} = \left(\myVM{R}{ru}{} \right)_{rs}$. 
\endgroup
\end{theorem}
\begin{proof}
See App.~\ref{App: Derivation of mean SNR} for the derivation of \eqref{Eq: Mean SNR Final}.
\end{proof}

Note that $F$ is the only variable dependent on the correlations in $\myVM{h}{ru}{}$ and also note that the variable $L_{rs}$ is a double integral of the loss function.  In Sec.~\ref{SubSec: No Correlation Lrs} and Sec.~\ref{SubSec: Full Correlation Lrs}, we derive exact results for special cases of $F$, $L_{rs}$ when $\Abs{\rho_{rs}}{} \in \{0,1\}$. These correlation extremes provide useful benchmarks to evaluate the SNR trends.

\subsection{Special Case 1: Uncorrelated $\hru{}$}\label{SubSec: No Correlation Lrs}
From \eqref{Eq: Lrs integral}, when $\hru{}$ is uncorrelated then $\phi_{r}$ and $\phi_{s}$ are i.i.d for $r \neq s$. Hence,
\begin{align}\label{Eq: Lrs when rhors is zero}
    \Exp{L(\phi_{r})L(\phi_{s})}
    = \left(\Exp{L(\phi_{r})}\right)^2 
    = \mu_1^2.
\end{align}
No correlation in $\hru{}$ also implies that $\rho_{mn} = 0$ for all $m \neq n$. Using this result, \eqref{Eq: Lrs when rhors is zero} and \cite[Eq.~(10)]{ISinghRayleigh},  $F$ simplifies to
\begin{align}\label{Eq: F when rhors is zero}
    F_u = \frac{\mu_1^2 N(N-1) \pi}{4}.
\end{align}
Therefore, the mean SNR for an uncorrelated $\hru{}$ channel is,
\begin{align}
    \Exp{\mathrm{SNR}}
    &=
    \bar{\tau}\Big( \betad M + 
    \sqrt{\betabr\betad\betaru}\Norm{\myVM{R}{d}{1/2}\ab{}}{}{}N\mu_1\frac{\pi}{2} \notag \\
    & \quad + 
    \betaru \betabr M(N\mu_2 + F_u) \Big) \label{Eq: Mean SNR Final uncorrelated}.
\end{align}
Note that the mean SNR expression depends on the PDL solely through the simple functions $\mu_1$ and $\mu_2$.
\subsection{Special Case 1: Perfect Correlation in $\hru{}$}\label{SubSec: Full Correlation Lrs}
With perfect correlation in $\hru{}$, $\rho_{rs}=1$ for $r,s=1,\ldots,N$. Hence, from \cite[Eq.~(13)]{ISinghRayleigh}, $F$ can be rewritten as  
$
    F =
    \underset{r \neq s}{\sum_{r=1}^{N} \sum_{s=1}^{N}} L_{rs} 
$
Under perfect correlation, we can exactly compute $\Exp{L(\phi_{r})L(\phi_{s})}$. Following App.~\ref{App: Derivation of mean SNR}, we can express the $i^{\mathrm{th}}$ RIS phase as $\phi_i = \angle\ab{H}\hd{} + \angle(\ar{})_i - \angle h_{\mathrm{ru},i}$. Hence, 
\begingroup
\allowdisplaybreaks
\begin{align*}
    &\Exp{L(\phi_{r})L(\phi_{s})} \notag \\
    &=
    \mathbb{E}\Big\{L(\angle\ab{H}\hd{} + \angle(\ar{})_r - \angle h_{\mathrm{ru},r}) \notag \\
    &\hspace{3em} \times
    L(\angle\ab{H}\hd{} + \angle(\ar{})_s - \angle h_{\mathrm{ru},s}) \Big\} \notag \\
    &=
    \Exp{L(w + \angle(\ar{})_r) L(w + \angle(\ar{})_s)} \notag \\
    &=
    \frac{1}{2\pi}\int_{0}^{2\pi} 
    L(w + \angle(\ar{})_r) L(w + \angle(\ar{})_s) \ dw.
\end{align*}
\endgroup
Using App.~\ref{App: Fancy integral}, the solution to the above integral is
\begingroup
\allowdisplaybreaks
\begin{align}
    &F_c 
    =  
    N(N-1)\left( \frac{A_1 A_2 2^{2\alpha + 1}}{\pi} \mathrm{B}\left(\frac{2\alpha+1}{2},\frac{2\alpha+1}{2} \right)
    +
    A_2^2 \right) \notag \\
    & + 
    \underset{r \neq s}{\sum_{r=1}^{N} \sum_{s=1}^{N}} \frac{A_1^2 2^{-2\alpha-1}}{\pi \sin(2\pi\alpha)}
    \begin{bmatrix}
        \sin(2\pi\alpha) & 1-\cos(2\pi\alpha) 
    \end{bmatrix}
    \begin{bmatrix}
        \Re\{\mathcal{I}\} \\ 
        \Im\{\mathcal{I}\}
    \end{bmatrix}  \label{Eq: F when rhors is 1}
\end{align}
\endgroup
with
\begin{equation}
    \mathcal{I} = 2 \pi (\gamma^2 - 1)^\alpha {}_{2}F_{1}(-2\alpha,2\alpha+1;1;\frac{1-\gamma_1}{2}),
\end{equation}
where $\gamma_1 = \gamma/\sqrt{\gamma^2 - 1}$ and $\gamma = \cos\left(\frac{\angle(\ar{})_r - \angle(\ar{})_s}{2}\right)$, $A_1 = 1 - L_{\mathrm{min}}$ and $A_2 = L_{\mathrm{min}}$.
Therefore, the mean SNR for a fully correlated channel is,
\begin{align}
    \Exp{\mathrm{SNR}}
    &=
    \bar{\tau}\Big( \betad M + 
    \sqrt{\betabr\betad\betaru}\Norm{\myVM{R}{d}{1/2}\ab{}}{}{}N\mu_1\frac{\pi}{2} \notag \\
    & \quad + 
    \betaru \betabr M(N\mu_2 + F_c) \Big) \label{Eq: Mean SNR Final fully correlated}.
\end{align}

\section{Impact of Loss Function on the mean SNR}\label{Sec: Impact of Phase Loss on mean SNR}
In this section, we explore the impact of the circuit-dependent parameters $L_{\mathrm{min}},\alpha,\theta$ on the mean SNR. These parameters only impact the variables $\mu_{1}, \mu_{2}, L_{rs}$ in the mean SNR expression \eqref{Eq: Mean SNR Final}. While the broad impact of $L_{\mathrm{min}},\alpha,\theta$ is intuitive from the loss function \eqref{Eq: Phase Loss Function}, in this section we present analysis to support and quantify these effects. 

\subsection{Phase Shift of the PDL Function: $\theta$}\label{SubSec: impact of theta on mean SNR}
The parameter, $\theta$, which controls the midpoint position of the loss function does not affect the mean SNR as $\Exp{L(\phi_r)}$ and $\Exp{L(\phi_r)L(\phi_s)}$ are averaged over an entire $2\pi$ period. Therefore, the mean SNR is independent of $\theta$.

\subsection{Steepness of the PDL Function: $\alpha$}\label{SubSec: impact of alpha on mean SNR}
The parameter $\alpha$ only affects the beta functions in $\mu_{1}, \mu_{2}$ and $L_{rs}$. From \cite[Eq.~(8.384.4)]{GradRyz}, we have
\begin{equation}\label{Eq: beta function general 1}
    \mathrm{B}\left( x, x \right)
    =
    2^{1 - 2x} \mathrm{B}\left( 1/2, x \right),
\end{equation}
which is a useful result as it appears in both $\mu_1$ and $\mu_2$.
Firstly, note that the series representation of \eqref{Eq: beta function general 1} given in \cite[Eq.~(8.382.3)]{GradRyz} shows that $\mathrm{B}\left(1/2, x \right)$ decreases in value as $x \to \infty$. Therefore, $\mathrm{B}\left( \frac{2\alpha+1}{2}, \frac{2\alpha+1}{2}\right)$ and $\mathrm{B}\left( \frac{4\alpha+1}{2}, \frac{4\alpha+1}{2}\right)$ are monotonically decreasing functions in $\alpha$ since $\alpha \geq 0$. Hence, from \eqref{Eq: mu1 final}-\eqref{Eq: mu2 final}, $\mu_1$ and $\mu_2$ benefit from having small $\alpha$. In terms of $L_{rs}$, note that $L_{rs}$ is a double integral over positive functions as $L(\phi_n) \in (0,1)$ and $g_{rs}$ is a positive function for all $v_{rs}(x)$ (see App.~\ref{App: g_rs is a positive function}). Therefore, $L_{rs}$ benefits from having small $\alpha$ since \eqref{Eq: Phase Loss Function} increases as $\alpha$ decreases. In summary, the mean SNR benefits from having a small $\alpha$ parameter.

\subsection{Minimum Amplitude of the PDL Function: $L_{\mathrm{min}}$}\label{SubSec: impact of Lmin on mean SNR}
 Let $c_1 = \frac{4^\alpha }{\pi} \mathrm{B}\left( \frac{2\alpha+1}{2}, \frac{2\alpha+1}{2}\right)$ and $c_2 = \frac{16^\alpha }{\pi} \mathrm{B}\left( \frac{4\alpha+1}{2}, \frac{4\alpha+1}{2}\right)$, then $\mu_1$ can be rewritten as
\begin{equation}\label{Eq: alternate mu_1}
    \mu_1 = 
    c_1
    +
    L_{\mathrm{min}} \left(2\pi - c_1 \right),
\end{equation}
and $L(\cdot) \in [0,1]$ implies that $L_{\mathrm{min}} \in [0,1]$. Using the results in App.~\ref{App: Alpha on beta func}, we can infer that for $\alpha \geq 0$, $2\pi > c_1$. Hence, \eqref{Eq: alternate mu_1} is an increasing function of $L_{\mathrm{min}}$. We can also rewrite $\mu_2$ as
\begin{equation}\label{Eq: alternate mu_2}
    \mu_2 = 
    L_{\mathrm{min}}^2 (2\pi + c_2 - 2c_1) + L_{\mathrm{min}}(2c_1 - 2c_2) + c_2.
\end{equation}
As above, we can use  App.~\ref{App: Alpha on beta func} to infer that $\alpha \geq 0$, $2\pi > c_2 - 2c_1$ and $c_1 > c_2$ so $\mu_2$ also increases with $L_{\mathrm{min}} \in [0,1]$. 
In terms of $L_{rs}$, note that $L_{rs}$ is a double integral over positive functions as $L(\phi_n) \in (0,1)$ and $g_{rs}$ is a positive function for all $v_{rs}(x)$ (see App.~\ref{App: g_rs is a positive function}). Therefore, $L_{rs}$ benefits from having large $L_{\mathrm{min}}$ since \eqref{Eq: Phase Loss Function} increases in value as $L_{\mathrm{min}}$ increases. In summary, the mean SNR benefits from high values of $L_{\mathrm{min}}$.

\section{Results}\label{Sec: Results}
We present numerical results to verify the analysis in Sec.~\ref{Sec: Impact of Phase Loss on mean SNR}. Firstly, note that we do not consider cell-wide averaging as the focus is on the SNR distribution over the fast fading. 
Hence, we present numerical results for fixed link gains for which the geometric model for the deployment of the UE, BS and RIS is adopted from \cite{ParameterValues} and shown in Fig.~\ref{Fig: BS-RIS-UE deployment}. 
\begin{figure}[h]
	\centering
	\includegraphics[width=5cm]{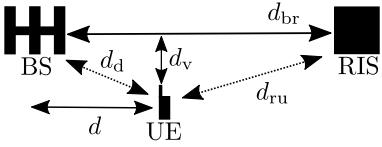}
	\raisecapt\caption{Deployment of BS, RIS and UE adapted from \cite{ParameterValues}.}
	\label{Fig: BS-RIS-UE deployment}
\end{figure}
In particular, since the RIS-BS link is LOS, we assume $\beta_{\mathrm{br}}=d_{\mathrm{br}}^{-2}$ where $d_{\mathrm{br}}=51$m. For the other channels, we use the distance-dependent path loss model,
\begin{equation}
    \beta_{\mathrm{ru}} = C_0 d_{\mathrm{ru}}^{-\alpha_{\mathrm{ru}}} 
    \quad , \quad
    \beta_{\mathrm{d}} = C_0 d_{\mathrm{d}}^{-\alpha_{\mathrm{d}}}, 
\end{equation}
where $C_0=-30$ dB is the path loss at a reference distance of 1m, $d_{\mathrm{ru}}=21.0238$m and $d_{\mathrm{d}}=30.167$m is the UE-RIS and UE-BS separation distances respectively, $\alpha_{\mathrm{ru}}=2.8$ and $\alpha_{\mathrm{d}}=3.5$ are the path loss exponents for the UE-RIS and UE-BS channels respectively. These values give the path gains of $\beta_{\mathrm{d}} = -81.7077$ dB and $\beta_{\mathrm{ru}} = -67.0360$ dB. Distances $d_{\mathrm{ru}}$ and $d_{\mathrm{d}}$ were computed using elementary trigonometry where $d=30$m and $d_{\mathrm{v}}=1$m.
The power of the transmitted signal is $E_s = 1$ and the noise power is $\sigma^2=-65$ dBm.

As stated in Sec. \ref{Sec: Channel Model}, the steering vectors for $\myVM{H}{br}{}$ are not restricted to any particular structure. However, for simulation purposes, we will use the VURA model as outlined in \cite{CMiller}, but in the $y-z$ plane with equal spacing in both dimensions at both the RIS and BS. The $y$ and $z$ components of the steering vector at the BS are $\myVM{a}{b,y}{} $ and $\myVM{a}{b,z}{}$ which are given by
\begin{align*}
& [1, e^{j2\pi d_{\mathrm{b}} \sin(\theta_{\mathrm{A}})\sin(\omega_{\mathrm{A}})}, \ldots, e^{j2\pi d_{\mathrm{b}} (M_{y}-1) \sin(\theta_{\mathrm{A}})\sin(\omega_{\mathrm{A}})}]^{T} \\
& \text{and } [1, e^{j2\pi d_{\mathrm{b}} \cos(\theta_{\mathrm{A}})}, \ldots, e^{j2\pi d_{\mathrm{b}}(M_{z}-1) \cos(\theta_{\mathrm{A}})}]^{T},
\end{align*}
respectively. Similarly at the RIS, $\myVM{a}{r,y}{}$ and $\myVM{a}{r,z}{}$ are defined by,
\begin{align*}
& [1, e^{j2\pi d_{\mathrm{r}} \sin(\theta_{\mathrm{D}})\sin(\omega_{\mathrm{D}})}, \ldots, e^{j2\pi d_{\mathrm{r}}(N_{y}-1) \sin(\theta_{\mathrm{D}})\sin(\omega_{\mathrm{D}})}]^{T} \\
& \text{and } [1, e^{j2\pi d_{\mathrm{r}} \cos(\theta_{\mathrm{D}})}, \ldots, e^{j2\pi d_{\mathrm{r}} (N_{z}-1) \cos(\theta_{\mathrm{D}})}]^{T},
\end{align*}
respectively, where $M = M_{y}M_{z}$, $N = N_{y}N_{z}$ with $M_{y}, M_{z}$ being the number of antenna columns and rows at the BS and $N_{y},N_{z}$ being the number of columns and rows of RIS elements. $d_{\mathrm{b}}=0.5$ and $d_{\mathrm{r}}$ are BS/RIS element spacings in wavelength units. Note that the value of $d_{\mathrm{r}}$ is set to satisfy a particular correlation level $\mathrm{sinc}(2 d_{\mathrm{r}}) = \rho_{\mathrm{ru}}$ as per \eq{\ref{Eq: Correlation Model}}.
The steering vectors at the BS and RIS are then given by,
\begin{equation}
\myVM{a}{b}{} = \myVM{a}{b,y}{} \otimes \myVM{a}{b,z}{}
\quad,\quad
\myVM{a}{r}{} = \myVM{a}{r,y}{} \otimes \myVM{a}{r,z}{},
\end{equation} 
respectively. $\theta_{\mathrm{A}}$ and $\omega_{\mathrm{A}}$ are elevation/azimuth angles of arrival (AOAs) at the BS and $\theta_{\mathrm{D}},\omega_{\mathrm{D}}$ are the corresponding angles of departure (AODs) at the RIS. The elevation/azimuth angles are selected based on the following geometry representing a range of LOS $\myVM{H}{br}{}$ links with less elevation variation than azimuth variation: 
$	
\theta_{D} \sim \mathcal{U}[70^{o},90^{o}], \hspace{1em}
\omega_{D} \sim \mathcal{U}[-30^{o},30^{o}], \hspace{1em}  
\theta_{A} = 180^{o} - \theta_{D}, \hspace{1em} 
\omega_{A} \sim \mathcal{U}[-30^{o},30^{o}] 
$. For all results in this paper we use a single sample from this range of angles given by $\theta_{D}=77.1^{o}, \omega_{D}=19.95^{o}, \theta_{A}=109.9^{o}, \omega_{A}=-29.9^{o}$.

In Fig.~\ref{Fig: figure 1}, we verify the mean SNR expression in \eqref{Eq: Mean SNR Final} for varying values of $N$ and $L_{\mathrm{min}}=0.5,\alpha=1.2,\theta=0.2, \rho_{\mathrm{d}}=0.7, \rho_{\mathrm{ru}} \in \{0,0.95,1\}$. For the special cases of $\rho_{\mathrm{ru}}=0$ and $\rho_{\mathrm{ru}}=1$, we use the expressions \eqref{Eq: F when rhors is zero} and \eqref{Eq: F when rhors is 1} to compute the variable $F$, respectively. For all correlation and RIS size scenarios, the analytical mean SNR agrees with simulations. {{Notice that even with PDL, the mean SNR grows as $\mathcal{O}(N^2)$, identical to the growth of the mean SNR without PDL \cite{Rayleigh}.}} Also shown is the trivial approximation where every element, $L(\phi_n)$, of the loss matrix is replaced by the mean of the loss function, $\mu_1$. As can be seen, this simple amplitude scaling gives a reasonable lower bound and is a useful rule of thumb.
\begin{figure}[h]
	\centering
	\myincludegraphics{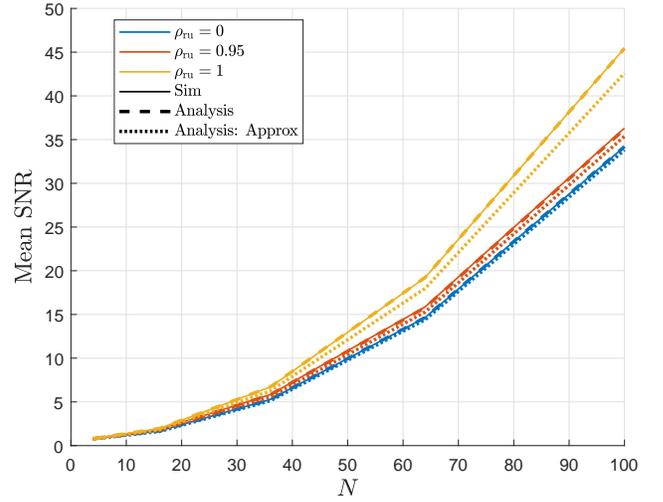}
	\raisecapt\caption{Simulated and analytical mean SNR vs $N$ for $L_{\mathrm{min}}=0.5,\alpha=1.2,\theta=0.2, \rho_{\mathrm{d}}=0.7, \rho_{\mathrm{ru}} \in \{0,0.95,1\}$.}
	\label{Fig: figure 1}
\end{figure}

In Fig.~\ref{Fig: figure 2}, we compute the analytical and simulated mean SNR for varying values of $\alpha$, and $L_{\mathrm{min}} \in \{0.1,0.5,0.95\}, N \in \{16,64\}, \theta \in \{0.2,0.42\}, \rho_{\mathrm{d}}=\rho_{\mathrm{ru}}=0.7$. 
\begin{figure}[h]
	\centering
	\myincludegraphics{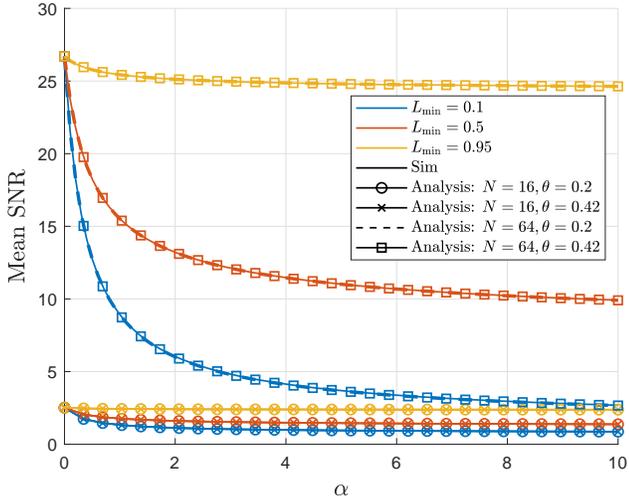}
	\raisecapt\caption{Simulated and analytical mean SNR for $L_{\mathrm{min}}=\{0.1,0.5,0.95\},\theta=\{0.2,0.42\}$.}
	\label{Fig: figure 2}
\end{figure}
It is observed in Fig.~\ref{Fig: figure 2} that the mean SNR monotonically decreases in $\alpha$ which agrees with the analysis in Sec.~\ref{Sec: Impact of Phase Loss on mean SNR}. Increasing the value of $L_{\mathrm{min}}$ increases the mean SNR which also agrees with the analysis. Furthermore, notice that as $L_{\mathrm{min}} \to 1$, the mean SNR becomes nearly constant for all values of $\alpha$ which agrees with analysis in that no loss is observed at $L_{\mathrm{min}} = 1$. As $\alpha \to 0$, the mean SNR converges to the same value for all $L_{\mathrm{min}}$ scenarios as per the analysis. Furthermore, notice that for both $N=16$ and $N=64$ scenarios, the mean SNR curves are identical for both $\theta=0.2$ and $\theta=0.42$. Hence, offsetting the loss function \eqref{Eq: Phase Loss Function} by $\theta$ does not affect the mean SNR as shown in the analysis

Next, we further examine how $\alpha$ and $L_{\mathrm{min}}$ affect the mean SNR by considering different RIS sizes. In Fig.~\ref{Fig: figure 2}, the initial drop off in mean SNR is steeper for the $N=64$ scenario compared to the scenario where $N=16$. Hence, as the number of RIS elements increases, the initial drop off in mean SNR is more pronounced. Also notice in Fig.~\ref{Fig: figure 2} that when $N=16$, the separation gap between the mean SNR curves for the three $L_{\mathrm{min}}$ values is smaller than those in the case of $N=64$. Therefore, as the number of RIS elements increases, altering $L_{\mathrm{min}}$ has a greater effect on the mean SNR.  
Typical parameter values $L_{\mathrm{min}}=0.2,\alpha=1.6$ are given in \cite{PracPhaseShift}. From Fig.~\ref{Fig: figure 2}, we see that the drop in SNR for these parameters and $N=64$ is bracketed by the $\alpha \in \{0.1,0.5\}$ curves and is between 48\% and 74\%. Hence we can expect a significant reduction in mean SNR for practical RIS systems.

\section{Conclusion}\label{Sec: Conclusion} 
{{We derive an exact closed form expression for the mean SNR where the RIS elements experience PDL. Specifically, the amplitude of the reflections from the RIS element are dependent on the optimal RIS phases which maximize the SNR in the absence of PDL. The attenuation function used for PDL was developed from a detailed circuit analysis of RIS elements, and is dependent on three parameters which control the minimum amplitude, steepness and shift of the attenuation function. We analytically characterise the impact of PDL on the mean SNR, offering insight into how PDL impacts the mean SNR performance. The analysis shows that the mean SNR only depends on the minimum amplitude and the steepness parameters. Having a larger minimum amplitude increases the mean SNR and having a steeper attenuation function decreases the mean SNR. This effect is enhanced when the number of RIS elements increases.}}

 

\begin{appendices}
\section{Derivation of mean SNR}\label{App: Derivation of mean SNR}
For ease of notation, we define the three terms in the SNR expression \eqref{Eq: SNR Eq} by
$
	\text{SNR} \triangleq \bar{\tau}\left( S_1 + S_2 + S_3 \right).
$
We then compute $\Exp{\text{SNR}}$ by considering each term in the expression.

\noindent \textbf{Term 1}:
Using \cite[Eq.~(52)]{RayleighRicean} and \eqref{Eq: channels hd hru}, we have
\begin{equation}
    \Exp{S_1} = \betad M \label{Eq: Mean S1}.
\end{equation}

\noindent \textbf{Term 2}:
Substituting the optimal RIS matrix \eqref{Eq: Optimum PHI} and the channels $\hd{},\Hbr{},\hru{}$ from Sec.~\ref{Sec: Channel Model} into $S_2$,
\begin{align}\label{App: Term 2 part 1}
    &\Exp{S_2} = 2\sqrt{\betabr} \notag \\
    &\times \Re\left\{ \Exp{
    \hd{H}\ab{}\ar{H}\psi\Diag{\ar{}} \Diag{e^{-j\angle\hru{}}} 
    L\left( \myVM{\Phi}{}{} \right) \hru{}
    } \right\} \notag \\
    &=
    2\sqrt{\betabr}\Re\left\{ \Exp{
    \Abs{\ab{H}\hd{}}{} \mathbf{1}_{N}^{T} L\left( \myVM{\Phi}{}{} \right) \Diag{e^{-j\angle\hru{}}} \hru{}
    } \right\} \notag \\
    &=
    2\sqrt{\betabr}\Re\left\{ \Exp{
    \Abs{\ab{H}\hd{}}{} \mathbf{1}_{N}^{T} L\left( \myVM{\Phi}{}{} \right) \Abs{\hru{}}{}
    } \right\}.
\end{align}
The matrix $L\left( \myVM{\Phi}{}{}\right)$ 
depends on $e^{j\angle \ab{H}\hd{}}$ and $e^{-j\angle\hru{}}$. Hence, 
\begin{align}
    &\Exp{S_2} = 
     2\sqrt{\betabr}\sum_{r=1}^{N} \Exp{\Abs{\ab{H}\hd{}}{}} \Exp{L(\phi_{r})} \Exp{\Abs{h_{\mathrm{ru},r}}{}},\notag
\end{align}
which is obtained by realising that $L(\phi_{r})$ is independent of both 
$\Abs{\ab{H}\hd{}}{}$ and $\Abs{h_{\mathrm{ru},r}}{}$. Noting that $\phi_r = \angle\ab{H}\hd{} + \angle(\ar{})_r - \angle h_{\mathrm{ru},r} \sim \mathcal{U}[0,2\pi]$, it follows that
\begin{equation}\label{Eq: Mean single generic loss function}
    \Exp{L(\phi_{r})} = \frac{1}{2\pi}\int_{0}^{2\pi} L(x) \ dx.
\end{equation}
Note that \eqref{Eq: Mean single generic loss function} is a generic calculation for any loss function. For the loss function given by \eqref{Eq: Phase Loss Function},
\begingroup
\allowdisplaybreaks
\begin{align}
    \Exp{L(\phi_{r})}
    &=
    \frac{1-L_{\mathrm{min}}}{2^\alpha 2\pi} \int_{0}^{2\pi}
    (1 + \sin(x + \theta))^\alpha \ dx + L_{\mathrm{min}} \notag \\
    &\overset{(a)}{=}
    \frac{4^\alpha(1-L_{\mathrm{min}})}{\pi}\mathrm{B}\left( \frac{2\alpha+1}{2}, \frac{2\alpha+1}{2} \right) + L_{\mathrm{min}} \notag \\
    &\triangleq \mu_1 \label{Eq: App mu1},
\end{align}
\endgroup
where $(a)$ uses App.~\ref{App: Reused sine Integral} to evaluate the integral. To complete the solution for $\Exp{S_2}$, we need to compute $\Exp{\Abs{\ab{H}\hd{}}{}}$ and 
$\sum_{r=1}^{N} \Exp{\Abs{h_{\mathrm{ru},r}}{}}$ which can be computed exactly using \cite[Eq.~(22)]{ISinghRayleigh}. Hence,
\begin{equation}
    \Exp{S_2}
    =
    \sqrt{\betabr\betad\betaru}\Norm{\myVM{R}{d}{1/2}\ab{}}{}{}N\mu_1\frac{\pi}{2} \label{Eq: Mean S2}.
\end{equation}
\noindent \textbf{Term 3}:
Substituting the optimal RIS matrix, \eqref{Eq: Optimum PHI}, and the channels $\Hbr{},\hru{}$ from Sec.~\ref{Sec: Channel Model} into $S_3$,
\begingroup
\allowdisplaybreaks
\begin{align}
    &\Exp{S_3}
    = \betabr
    \mathbb{E}\Big\{
    \hru{H} L\left( \myVM{\Phi}{}{} \right)\frac{\hd{H}\ab{}}{\Abs{\ab{H}\hd{}}{}} \Diag{e^{j\angle\hru{}}} \Diag{\ar{H}}
    \ar{}   \notag \\ 
    & \qquad \times \ab{H} \ab{} \ar{H} \frac{\ab{H}\hd{}}{\Abs{\ab{H}\hd{}}{}} \Diag{\ar{}} \Diag{e^{-j\angle\hru{}}} L\left( \myVM{\Phi}{}{} \right) \hru{}
    \Big\} \notag \\
    &= \betabr M 
    \mathbb{E}\Big\{
    \hru{H} L\left( \myVM{\Phi}{}{} \right) \Diag{e^{j\angle\hru{}}} 
    \mathbf{1}_N \mathbf{1}_N^T \Diag{e^{-j\angle\hru{}}} \notag \\
    & \qquad \times L\left( \myVM{\Phi}{}{} \right)
    \hru{}
    \Big\} \notag \\
    &=
    \betabr M \sum_{r=1}^{N}\sum_{s=1}^{N} \mathbb{E}\Big\{
    h_{\mathrm{ru},r}^* L(\phi_{r}) e^{j\angle h_{\mathrm{ru},r}}
    e^{-j\angle h_{\mathrm{ru},s}} L(\phi_{s}) h_{\mathrm{ru},s}
    \Big\} \notag \\
    &= \betabr M 
    \Bigg( \sum_{r=1}^{N} \Exp{\Abs{h_{\mathrm{ru},r}}{2} } \Exp{L^2(\phi_{r})} \notag \\
    & \qquad + \underset{r \neq s}{\sum_{r=1}^{N} \sum_{s=1}^{N}}
    \Exp{\Abs{h_{\mathrm{ru},r}}{}\Abs{h_{\mathrm{ru},s}}{} } 
    \Exp{L(\phi_{r})L(\phi_{s})} \Bigg) \label{App: Term 3 part 1}. 
\end{align}
\endgroup
The first term in \eqref{App: Term 3 part 1} requires $\Exp{\Abs{h_{\mathrm{ru},r}}{2}} = \betaru$. To obtain $\Exp{L^2(\phi_{r})}$,  
we expand the square of \eqref{Eq: Phase Loss Function},
\begingroup
\allowdisplaybreaks
\begin{align}
    L^2(\phi_{r})
    &=
    (1-L_{\mathrm{min}})^2 2^{-2\alpha} 
    \left( {\sin(\phi_r + \theta) + 1} \right)^{2\alpha} + L^2_{\mathrm{min}} \notag \\
    & \quad + 2^{-\alpha-1}L_{\mathrm{min}}(1-L_{\mathrm{min}})
    \left( {\sin(\phi_r + \theta) + 1} \right)^{\alpha} \notag \\
    &\triangleq L_1 + L_2 + L_3.
\end{align}
\endgroup
The mean of the first term is
\begingroup
\allowdisplaybreaks
\begin{align}
    \Exp{L_1} 
    &= 
    \frac{(1 - L_{\mathrm{min}})^2}{2 \pi 4^{\alpha}}\int_{0}^{2\pi} (1 + \sin(x + \theta))^{2\alpha} \ dx \notag \\
    &\overset{(a)}{=}
    \frac{(1 - L_{\mathrm{min}})^2 16^\alpha}{\pi} \mathrm{B}\left( \frac{4\alpha+1}{2}, \frac{4\alpha+1}{2} \right) \label{Eq: L_1},
\end{align}
\endgroup
where, in $(a)$, App.~\ref{App: Reused sine Integral} is used to evaluate the integral. 
The mean of the second term is simply $\Exp{L_2} = L^2_{\mathrm{min}}$.
The mean of the third term is,
\begin{align}
    \Exp{L_3}
    &\overset{(a)}{=}
    \frac{L_{\mathrm{min}}(1 - L_{\mathrm{min}}) 2^{2\alpha+1}}{\pi} \mathrm{B}\left( \frac{2\alpha+1}{2}, \frac{2\alpha+1}{2} \right) \label{Eq: L_3},
\end{align}
where, in $(a)$, App.~\ref{App: Reused sine Integral} is used to evaluate the integral. 
Summing the three expectations, we have
\begingroup
\allowdisplaybreaks
\begin{align}
    \Exp{L^2(\phi_{r})}
    &=
    \frac{L_{\mathrm{min}}(1 - L_{\mathrm{min}}) 2^{2\alpha+1}}{\pi} \mathrm{B}\left( \frac{2\alpha+1}{2}, \frac{2\alpha+1}{2} \right) \notag \\
    & +
    L^2_{\mathrm{min}}
    +
    \frac{(1 - L_{\mathrm{min}})^2 16^\alpha}{\pi} \mathrm{B}\left( \frac{4\alpha+1}{2}, \frac{4\alpha+1}{2} \right) \notag \\
    &\triangleq \mu_2 \label{Eq: mean L2}.
\end{align}
\endgroup

The second term in \eqref{App: Term 3 part 1} requires $\Exp{\Abs{h_{\mathrm{ru},r}}{}\Abs{h_{\mathrm{ru},s}}{} } $ and $\Exp{L(\phi_{r})L(\phi_{s})}$. Using \cite[Eq.~(11)]{Shuang} we have
\begin{equation}
    \Exp{\Abs{h_{\mathrm{ru},r}}{}\Abs{h_{\mathrm{ru},s}}{} }
    =
    \dfrac{\pi}{4}\left( 1 - \left\lvert\rho_{ik}\right\rvert^2\right)^2 {}_{2}F_{1}\left(\frac{3}{2},\frac{3}{2};1;\left\lvert\rho_{ik}\right\rvert^2 \right) \label{Eq: Part hypergeometric},
\end{equation}
where ${}_{2}F_{1}(\cdot)$ is the Gaussian hypergeometric function and $\rho_{ij} = \left(\myVM{R}{ru}{} \right)_{ij}$. 

The final expectation required is $\Exp{L(\phi_{r})L(\phi_{s})}$. Let $x = \angle h_{\mathrm{ru},r}, y = \angle h_{\mathrm{ru},s}$, then the joint density of phases $x,y$ is given by \cite[Eq.~(3.12)]{KSMiller},
\begingroup
\allowdisplaybreaks
\begin{align}
    f_{X,Y}(x,y) &= \frac{1 - \Abs{\rho_{rs}}{2}}{8 \pi^2 } \frac{\partial^2}{\partial\lambda^2}(\cos^{-1}(\lambda))^2 \notag \\
    &= - \frac{1 - \Abs{\rho_{rs}}{2}}{4 \pi^2 } \frac{\partial}{\partial\lambda}\left( \cos^{-1}(\lambda) (1-\lambda^2)^{-1/2} \right) \notag \\
    &=
    \frac{1 - \Abs{\rho_{rs}}{2}}{4 \pi^2 } \left(
    \frac{1}{1 - \lambda^2} - \frac{\lambda \cos^{-1}(\lambda)}{(1 - \lambda^2)^{3/2}}
    \right) \notag \\
    &\triangleq g(x-y), \label{Eq: g(x-y)}
\end{align}
\endgroup
with 
\begin{align}
    \lambda &= \Abs{\rho_{rs}}{} \cos(x-y-\angle(-\rho_{rs})).
\end{align}
Recall that each optimal RIS phase is $\phi_r = \angle\ab{H}\hd{} + \angle(\ar{})_r - \angle h_{\mathrm{ru},r}$. To obtain the joint density of $\phi_{r},\phi_{s}$ defined by $f_{r,s}(x,y)$, let $Z = \angle\ab{H}\hd{}$, $a = \angle(\ar{})_r$ and $b = \angle(\ar{})_s$. Then conditioned on $Z=z$, we have the conditional PDF
\begingroup
\allowdisplaybreaks
\begin{align*}
    f_{r,s|Z}(u,v \lvert z) &= f_{X,Y}(z+a-u,z+b-v) \notag \\
    &= g(v-u+a-b),
\end{align*}
\endgroup
where the domain of $u$ and $v$ is $z+a-2\pi \leq u \leq z+a$ and $z+b-2\pi \leq v \leq z+v$ respectively. This gives,
\begingroup
\allowdisplaybreaks
\begin{align}
    &\Exp{L(\phi_{r})L(\phi_{s})}
    =
    \mathbb{E}\left\{ \mathbb{E}\left\{ L(\phi_{r})L(\phi_{s}) \lvert Z \right\} \right\} \notag \\
    &=
    \int_{0}^{2\pi} \frac{1}{2\pi} 
    \int_{z+b-2\pi}^{z+b}
    \int_{z+a-2\pi}^{z+a}
    L(u) L(v) \notag \\
    & \qquad \times g(v-u+a-b) \ dudvdz \notag.
\end{align}
\endgroup
Let $s=u+2\pi-a$ and $t=v+2\pi-b$, then
\begingroup
\allowdisplaybreaks
\begin{align}
    &\Exp{L(\phi_{r})L(\phi_{s})}\notag \\
    &= 
    \int_{0}^{2\pi} \frac{1}{2\pi} 
    \int_{z}^{z+2\pi}
    \int_{z}^{z+2\pi}
    L(s+a-2\pi) L(t+b-2\pi) \notag \\
    & \qquad \times g(t-s) \ dsdtdz \notag \\
    &=
    \int_{0}^{2\pi}
    \int_{0}^{2\pi}
    L(s+a-2\pi) L(t+b-2\pi) g(t-s) \ ds dt \label{Eq: Joint expectation phase loss}.
\end{align}
\endgroup
Therefore, the second term of \eqref{App: Term 3 part 1} is given by
\begingroup
\allowdisplaybreaks
\begin{align}
    F &= \underset{r \neq s}{\sum_{r=1}^{N} \sum_{s=1}^{N}} \dfrac{\pi}{4}\left( 1 - \left\lvert\rho_{rs}\right\rvert^2\right)^2 {}_{2}F_{1}\left(\frac{3}{2},\frac{3}{2};1;\left\lvert\rho_{rs}\right\rvert^2 \right)
    \notag \\
    & \quad \times \int_{0}^{2\pi}
    \int_{0}^{2\pi}
    L(s+\angle(\ar{})_r-2\pi) L(t+\angle(\ar{})_s-2\pi) \notag \\
    & \quad \times g(t-s) \ ds dt,
\end{align}
\endgroup
which gives the expectation of the final term as,
\begin{align}
    \Exp{S_3}
    &=
    \betaru \betabr M(N\mu_2 + F) \label{Eq: Mean S3}.
\end{align}

Combining \eqref{Eq: Mean S1}, \eqref{Eq: Mean S2} and \eqref{Eq: Mean S3}, completes the derivation.

\section{$\int_{0}^{2\pi} (1+\sin(x + a))^b \ dx$}\label{App: Reused sine Integral}
Let $a \in \Re$ and $b \geq 0$. Then,
\begingroup
\allowdisplaybreaks
\begin{align}
    &\int_{0}^{2\pi} (1+\sin(x + a))^b \ dx 
    =
    \int_{0}^{2\pi} \left(2 \sin^2\left(\frac{x+a}{2} + \frac{\pi}{4}\right)\right)^b \ dx \notag \\
    &=
    2^{b+2} \int_{0}^{\pi/2} \sin^{2b}(x) \ dx =
    2^{3b+1} \mathrm{B}\left(\frac{2b+1}{2},\frac{2b+1}{2} \right),
\end{align}
\endgroup
where we use \cite[Eq.~(3.621.1)]{GradRyz} in the final step.

\section{$\int_{0}^{2\pi} L(x + a) L(x + b) \ dx$}\label{App: Fancy integral}
Let $A_1 = 1 - L_{\mathrm{min}}$ and $A_2 = L_{\mathrm{min}}$, then,
\begingroup
\allowdisplaybreaks
\begin{align*}
    &L(x + a) L(x + b) =
    A_1A_2 \left(\frac{\sin(x + a + \theta) + 1}{2}\right)^\alpha \notag \\
    &+ A_2^2 
    + A_1A_2 \left(\frac{\sin(x + b + \theta) + 1}{2}\right)^\alpha \notag \\
    & +
    A_1^2 \left(\frac{\sin(x + a + \theta) + 1}{2}\right)^\alpha
    \left(\frac{\sin(x + b + \theta) + 1}{2}\right)^\alpha \notag \\
    & \triangleq L_1 + L_2 + L_3 + L_4.
\end{align*}
\endgroup
Using App.~\ref{App: Reused sine Integral}, we can integrate the first term to obtain
\begin{equation}\label{App: L1}
    \int_{0}^{2\pi} L_1 \ dx = 
    A_1 A_2 2^{2\alpha + 1} \mathrm{B}\left(\frac{2\alpha+1}{2},\frac{2\alpha+1}{2} \right).
\end{equation}
The second term is
\begin{equation}\label{App: L2}
    \int_{0}^{2\pi} L_2 \ dx = 
    2 \pi A_2^2 .
\end{equation}
The third term can be computed using App.~\ref{App: Reused sine Integral} to obtain,
\begin{equation}\label{App: L3}
    \int_{0}^{2\pi} L_3 \ dx = 
    A_1 A_2 2^{2\alpha + 1} \mathrm{B}\left(\frac{2\alpha+1}{2},\frac{2\alpha+1}{2} \right).
\end{equation}
Integrating the fourth term requires more work.
\begingroup
\allowdisplaybreaks
\begin{align}
    &\int_{0}^{2\pi} L_4 \ dx \notag \\
    &=
    \frac{A_1^2}{4^\alpha} \int_{0}^{2\pi}
    (1+\sin(x + a + \theta))^\alpha (1+\sin(x + b + \theta))^\alpha \ dx \notag \\
    &= 
    A_1^2 \int_{0}^{2\pi}
    \left( \sin\left(\frac{x+a}{2}\right) 
    \sin\left(\frac{x+b}{2}\right) \right)^{2\alpha} \ dx \notag \\
    &=
    \frac{A_1^2}{4^\alpha} \int_{0}^{2\pi}
    \left( \gamma - \cos\left(x + \frac{a+b}{2} \right) \right) ^{2\alpha}
    \ dx \notag \\
    &=
    \frac{A_1^2}{4^\alpha} \int_{0}^{2\pi}
    \Abs{\gamma - \cos\left(x + \frac{a+b}{2} \right)}{2\alpha} 
    ( 1 + e^{j 2 \pi \alpha} ) \ dx \notag \\
    &= \frac{A_1^2}{4^\alpha} \left( \mathcal{I}_{\mathrm{real}} + \mathcal{I}_{\mathrm{imag}} \right), 
\end{align}
\endgroup
where $\gamma = \cos(\frac{a+b}{2})$, $\Re\{\mathcal{I}\} = \mathcal{I}_{\mathrm{real}} + \mathcal{I}_{\mathrm{imag}}\cos(2\pi\alpha)$ and $\Im\{\mathcal{I}\} = \mathcal{I}_{\mathrm{imag}}\sin(2\pi\alpha)$ with
\begingroup
\allowdisplaybreaks
\begin{align}
    \mathcal{I} 
    &= \int_{0}^{2\pi}
    \left( \gamma - \cos\left(x + \frac{a+b}{2} \right) \right) ^{2\alpha}
    \ dx \notag \\
    &\overset{(a)}{=}
    2 \pi (\gamma^2 - 1)^\alpha P_{2\alpha}(\gamma_1) \notag \\
    &\overset{(b)}{=}
    2 \pi (\gamma^2 - 1)^\alpha {}_{2}F_{1}\left(-2\alpha,2\alpha+1;1;\frac{1-\gamma_1}{2}\right),
\end{align}
\endgroup
where $\gamma_1 = \gamma/\sqrt{\gamma^2 - 1}$, $(a)$ uses \cite[Eq.~(3.661.3)]{GradRyz} to evaluate the integral and $(b)$ uses a hypergeometric transformation of the Legendre function for arbitrary degrees \cite[Eq.~(8.820.1)]{GradRyz}. Forming a system of linear equation with $\Re\{\mathcal{I}\}$ and $\Im\{\mathcal{I}\}$, the fourth integral is,
\begingroup
\allowdisplaybreaks
\begin{align}
    \int_{0}^{2\pi} L_4 \ dx 
    = 
    \frac{A_1^2 2^{-2\alpha}}{\sin(2\pi\alpha)}
    \begin{bmatrix}
        \sin(2\pi\alpha) & 1-\cos(2\pi\alpha) 
    \end{bmatrix}
    \begin{bmatrix}
        \Re\{\mathcal{I}\} \\ 
        \Im\{\mathcal{I}\}
    \end{bmatrix} \label{App: L4}.
\end{align}
\endgroup
Combining \eqref{App: L1}, \eqref{App: L2}, \eqref{App: L3} and \eqref{App: L4} completes the solution.

\section{$g_{rs}$ is a positive function}\label{App: g_rs is a positive function}
Here, we show that the function $g_{rs}$ given by \eqref{Eq: g_rs} is a positive function. Firstly, we can rewrite $g_{rs}$ as
\begin{equation}
    g_{rs}(x)
    =
    \gamma \left(
    \sqrt{1 - v_{rs}(x)^2} - v_{rs}(x) \cos^{-1}(v_{rs}(x))
    \right),
\end{equation}
where $\gamma = \frac{1-\Abs{\rho_{rs}}{2}}{4 \pi^2 (1 - v_{rs}(x)^2)^{3/2}}$. From \eqref{Eq: v_rs}, we know that $-1 \leq v_{rs}(x) \leq 1$ and since $0 \leq \Abs{\rho_{\mathrm{ru}}}{} \leq 1$, then $\gamma$ is always positive. On the region $-1 \leq v_{rs}(x) < 0$, $v_{rs}(x) \cos^{-1}(v_{rs}(x)) < 0$ since $ \cos^{-1}(v_{rs}(x)) \geq 0$. Hence $g_{rs}$ is positive over $-1 \leq v_{rs}(x) < 0$.

In the region $0 \leq v_{rs}(x) \leq 1$ and noting that $\sqrt{1 - v_{rs}^2(x)} = \sqrt{1 - v_{rs}(x)}\sqrt{1 + v_{rs}(x)}$, we can use \cite[Eq.~(5)]{InvCosIneq} to obtain the inequality,
\begin{equation}\label{Eq: inv cos inequality}
    \cos^{-1}(v_{rs}(x))
    < \left( 2 - \sqrt{\frac{2}{1+v_{rs}(x)}} \right)
    \frac{\sqrt{1 - v_{rs}^2(x)}}{v_{rs}(x)}.
\end{equation}
Since $0 \leq v_{rs}(x) \leq 1$, then we can rewrite  \eqref{Eq: inv cos inequality} as 
\begin{equation}
    v_{rs}(x)\cos^{-1}(v_{rs}(x))
    <
    \sqrt{1 - v_{rs}^2(x)}.
\end{equation}
Hence $g_{rs}$ is also positive over $0 < v_{rs}(x) \leq 1$.


\section{Impact of $\alpha$ on Beta functions}\label{App: Alpha on beta func}
Here, we mathematically characterise the impact of $\alpha$ on the beta functions in \eqref{Eq: mu1 final} and \eqref{Eq: mu2 final}. As $\alpha \to 0$, 
\begin{equation}\label{Eq: trend 1 alpha becomes zero}
    4^\alpha \mathrm{B}\left( \frac{2\alpha+1}{2}, \frac{2\alpha+1}{2}\right)
    \to
    \mathrm{B}\left( \frac{1}{2}, \frac{1}{2}\right)
    \overset{(a)}{=}
    \frac{\Gamma^2(1/2)}{\Gamma(1)}
    =
    \pi,
\end{equation}
and
\begin{equation}\label{Eq: trend 2 alpha becomes zero}
    16^\alpha \mathrm{B}\left( \frac{4\alpha+1}{2}, \frac{4\alpha+1}{2}\right)
    \to
    \mathrm{B}\left( \frac{1}{2}, \frac{1}{2}\right)
    =
    \pi,
\end{equation}
where $(a)$ uses \cite[Eq.~(8.384.1)]{GradRyz}. 

As $\alpha \to \infty$, we use the following asymptotic expressions
\begin{equation}\label{Eq: beta large alpha1}
    4^\alpha\mathrm{B}\left( \frac{2\alpha+1}{2}, \frac{2\alpha+1}{2}\right)
    =
    4^\alpha\frac{\Gamma^2(\alpha + 1/2)}{\Gamma(2\alpha + 1)}
    \overset{(b)}{\sim}
    \sqrt{\frac{\pi}{\alpha}},
\end{equation}
\begin{equation}\label{Eq: beta large alpha2}
    16^\alpha\mathrm{B}\left( \frac{4\alpha+1}{2}, \frac{4\alpha+1}{2}\right)
    =
    16^\alpha\frac{\Gamma^2(\alpha + 1/2)}{\Gamma(2\alpha + 1)}
    \overset{(b)}{\sim}
    \sqrt{\frac{\pi}{\alpha}},
\end{equation}
where $(b)$ uses \cite[Eq.~(6.1.39)]{Stegun} for asymptotic formulas for gamma functions. Clearly, from \eqref{Eq: beta large alpha1} and \eqref{Eq: beta large alpha2}, $4^\alpha \mathrm{B}\left( \frac{2\alpha+1}{2}, \frac{2\alpha+1}{2}\right) \to 0$ and $16^\alpha \mathrm{B}\left( \frac{4\alpha+1}{2}, \frac{4\alpha+1}{2}\right) \to 0$ as $\alpha \to \infty$.

\end{appendices}

\bibliographystyle{IEEEtran}
\bibliography{main.bib}

\end{document}